\documentclass[sigconf]{acmart}
\usepackage{bm}
\usepackage{amsthm}
\usepackage{algorithm, pseudocode}
\usepackage{mathrsfs}
\usepackage{enumitem}
\usepackage{todonotes}
\usepackage[titles]{tocloft}
\usepackage{xcolor}
    
\newtheorem{definition}{Definition}
\numberwithin{definition}{section}
\newtheorem{theorem}[definition]{Theorem}
\newtheorem{corollary}[definition]{Corollary}
\newtheorem{proposition}[definition]{Proposition}
\newtheorem{lemma}[definition]{Lemma}

\newtheorem{example}[definition]{Example}

\def\K{\ensuremath{\mathbb{Q}}}
\def\Kbar{\ensuremath{\overline{\mathbb{K}}}}
\def\Kbar{\ensuremath{{\mathbb{C}}}}
\def\fieldC{\ensuremath{\mathbb{C}}}
\def\field{\ensuremath{\mathbb{Q}}}
\def\R{\ensuremath{\mathbb{R}}}
\def\C{\ensuremath{\mathbb{C}}}
\def\Q{\ensuremath{\mathbb{Q}}}
\def\N{\ensuremath{\mathbb{N}}}
\def\T{\ensuremath{\mathbb{T}}}

\def\a{\bm{a}}
\def\c{\bm{c}}

\def\h{\bm{h}}

\def\vl{{L}}

\def\rank{\ensuremath{{\rm rank}}}
\def\jac{\ensuremath{{\rm Jac}}}

\def\fraka{{\mathfrak{a}}}

\def\Nmap{{\bm{P}}}

\def\calU{{\mathcal{U}}}
\def\calN{{\phi}}
\def\calA{{\mathcal{A}}}

\def\calS{{\mathcal{S}}}

\def\bmx{{\bm{z}}}
\def\bmn{{\bm{p}}}
\def\bmN{{\bm{P}}}

\def\bmH{{\bm{H}}}
\def\bmMV{{\bm{V}}}
\def\n{{p}}
\def\N{{P}}

\def\softO{\ensuremath{{O}{\,\tilde{ }\,}}}
\def\scrR{\ensuremath{\mathscr{R}}}

\def\macrodelta{t}

\copyrightyear{2023}
\acmYear{2023}
\setcopyright{acmlicensed}\acmConference[ISSAC '23]{Proceedings of the 2023 International Symposium on Symbolic and Algebraic Computation}{July 24--27, 2023}{Troms\o, Norway}
\acmPrice{15.00}
\acmDOI{??}
\acmISBN{??}
\title[Faster real root decision algorithm for symmetric
polynomials]{Faster real root decision algorithm for symmetric polynomials}

\author{George Labahn}
\affiliation{%
  \institution{Cheriton School of Computer Science}%
  \city{University of Waterloo}%
  \country{Ontario, Canada} }
   
  \author{Cordian Riener}
\affiliation{%
  \institution{Department of Mathematics and Statistics}
  \city{UiT, The Arctic University of Norway, Troms\o{}} 
  \postcode{N-9037}\country{Norway}
}

\author{Mohab Safey El Din}
\affiliation{%
	\institution{Sorbonne Universit\'e, CNRS, LIP6}
	\city{F-75005 Paris}
	\postcode{75252}\country{France}}

\author{\'Eric Schost}
\affiliation{%
  \institution{Cheriton School of Computer Science}%
  \city{University of Waterloo}%
  \country{Ontario, Canada} }

\author{Thi Xuan Vu}
\affiliation{%
  \institution{Department of Mathematics and Statistics}
  \city{UiT, The Arctic University of Norway, Troms\o{}} 
  \postcode{N-9037}\country{Norway}
}

\copyrightyear{2023}
\acmYear{2023}
\setcopyright{rightsretained}
\acmConference[ISSAC 2023]{International Symposium on Symbolic and Algebraic Computation 2023}{July 24--27, 2023}{Tromsø, Norway}
\acmBooktitle{International Symposium on Symbolic and Algebraic Computation 2023 (ISSAC 2023), July 24--27, 2023, Tromsø, Norway}\acmDOI{10.1145/3597066.3597097}
\acmISBN{979-8-4007-0039-2/23/07}

\begin{document}

\begin{abstract}
  In this paper, we consider the problem of deciding the existence of real
  solutions to a system of polynomial equations having real coefficients, and
  which are invariant under the action of the symmetric group. We construct and
  analyze a Monte Carlo probabilistic algorithm which solves this problem, under
  some regularity assumptions on the input, by taking advantage of the symmetry
  invariance property.

  The complexity of our algorithm is polynomial in $d^s, {{n+d} \choose d}$, and
  ${{n} \choose {s+1}}$, where $n$ is the number of variables and $d$ is the
  maximal degree of $s$ input polynomials defining the real algebraic set under
  study. In particular, this complexity is polynomial in $n$ when $d$ and $s$
  are fixed and is equal to $n^{O(1)}2^n$ when $d=n$.
\end{abstract}

\maketitle 
 
\section{Introduction}

Let $\bm{f}=(f_1, \dots, f_s)$ be polynomials in the multivariate
polynomial ring $\Q[x_1, \ldots, x_n]$ and let $V(\bm{f}) \subset \C^n$ be
the algebraic set defined by $\bm{f}$. We denote by $V_\R(\bm{f}) := V(\bm{f})
\cap \R^n$ the set of solutions in $\R^n$ to the system $\bm{f}$.  In
addition we assume that all $f_i$'s are invariant under the action of the
symmetric group $S_n$, that is, are symmetric polynomials (or
equivalently, $S_n$-invariant polynomials).

Under this invariance property, we design an algorithm
which, on input $\bm{f}$, decides whether $V_\R(\bm{f})$ is empty or not. As is typical for such
problems, we assume that the Jacobian matrix of $\bm{f}$ with respect to
$x_1, \dots, x_n$ has rank $s$ at any point of $V(\bm{f})$.  In this case
the Jacobian criterion \cite[Thm 16.19]{eisenbud2013commutative}
implies that the complex algebraic set $V(\bm{f})$ is smooth and
$(n-s)$-equidimensional (or empty).

\smallskip\noindent{{\bf Previous work.}}
The real root decision problem for polynomial systems of equations
(and more generally systems of inequalities) lies at the foundations
of computational real algebraic geometry.  Algorithms for solving
polynomial systems over the real numbers start with
Fourier~\cite{Fourier26} who provided a first algorithm for solving
linear systems of inequalities (rediscovered in 1919 by
Dines~\cite{Dines19}). These algorithms are important because they
make the first connection with elimination theory. Tarski's
theorem~\cite{Tarski} states that the projection of a semi-algebraic
set on a coordinate subspace is a semi-algebraic set. This theorem,
and its algorithmic counterpart which relies on Sturm's theorem for
real root counting in the univariate case, enable recursive
algorithmic patterns (eliminating variables one after another). The
first algorithm with an elementary recursive complexity, {\it
  Cylindrical Algebraic Decomposition}, is due to Collins (see
\cite{Collins} and references in~\cite{McCallum, McCallumeq, Hong:92d,
  Strz06, Strz14, Chen09, EnDa20, Chen20} for various further
improvements).

It turns out that these algorithms run in time doubly exponential in
$n$~\cite{Dav88, BrDa07}. Note that some variants actually solve the
quantifier elimination problem, a much more general and difficult
computational problem than the real root decision problem.

Algorithms which solve the real root decision problem in time singly
exponential in $n$ and polynomial in the maximum degree of the input
were pioneered by Grigoriev and Vorobjov~\cite{GV88}
and Renegar~\cite{Ren}, and further improved by Canny~\cite{Canny},
Heintz, Roy and Solern\'o~\cite{HRS93} and Basu, Pollack and
Roy~\cite{BPR96}. The method used in this framework is referred to as
the {\em critical point method}.  It reduces the real root decision
problem to the computation of finitely many complex critical points of
a polynomial map which reaches extrema at each connected component of
the semi-algebraic set under study.

The algorithm proposed here for solving the real root decision problem
for systems of 
symmetric  polynomial equations also builds on the
critical point method.
It borrows ideas from probabilistic algorithms which have been
designed to obtain sharper complexity estimates (e.g. cubic either in
some B\'ezout bound attached to some critical point system or in some
geometric intrinsic degree) and obtain practical performances that
reflect the complexity gains~\cite{BGHM1, BGHM2, BGHM3, BGHM4, BGHM5,
  SaSc03, BGHS14}. These algorithms make use of geometric resolution
or symbolic homotopy techniques to control the complexity of the
algebraic elimination step~(see e.g.~\cite{GiLeSa01, SaSc18} and
references therein), and of regularity assumptions to easily derive
critical point systems from the input polynomials.

Under the Jacobian criterion assumptions, critical points are defined
as the intersection of the affine variety $V(\bm{f})$ with a determinantal
variety derived from a certain Jacobian matrix. The design of
dedicated algebraic elimination algorithms for this particular setting
has attracted some attention already~\cite{FaSaSp12, BGHLM15, Spa14,
  SaSp16, HSSV21}. When adding the symmetry 
  property to
polynomials defining the variety {\em and} the polynomial map for
which one computes the critical points, significant improvements have
been achieved recently in~\cite{faugere2020computing} by using the
symbolic homotopy algorithms in~\cite{labahn2021homotopy}.

These improvements, which allows one to obtain complexity gains
related to the combinatorial complexity of the symmetric group, also
borrow ideas from algebraic algorithms working with data which are
invariant by the action of this group~\cite{FaugereJules12}. We
emphasize that taking advantage of symmetries in data is a topical and
difficult issue, which involves a variety of
methodologies~\cite{perminov2009, buse2016resultant, Colin97,
  Sturmfels2008, Fau09}.

In~\cite{Timofte2003}, Timofte proves a breakthrough result which is
now known as the degree principle. It states that 
a symmetric 
polynomial of degree $d$ with real coefficients has real solutions if
and only if one of these solutions has at most $d$ distinct
coordinates.

This shows that when $d$ is fixed and $n$ grows, the real root
decision problem can be solved in polynomial time. This is far better
than computing at least one sample point per connected component (see
also~\cite{BaRie17, basu2017efficient, basu2022vandermonde}), and is
one of the rare interesting cases where the best known algorithms for
these two problems admit different complexities.
This is also the starting point of several results which enhance the
real root decision problem and polynomial optimization under some
$S_n$-invariance property for classes of problems where $d$ remains
fixed and $n$ grows (see~\cite{Rie13, Riener2012, Riener2016,
  gatermann2004symmetry} and \cite{RieSa18} for equivariant systems).

\smallskip\noindent{{\bf Main contributions.}}
Being able to leverage $S_n$-invariance for critical point
computations is not sufficient to solve root decision problems more
efficiently using the critical point method. Additional techniques are
needed.

Indeed, to solve the real root decision problem by finding the
critical points of a polynomial map $\phi$, one typically defines
$\phi$ as the distance from points on the variety to a generic
point. This map reaches extrema at each connected component of the
semi-algebraic set under study. However, the map $\phi$ is not
symmetric. If it was, our problem would be solved by the critical
point algorithm of \cite{faugere2020computing}. Unfortunately there
does not appear to be an obvious symmetric map that fits the bill.

Instead, our approach is to apply the critical point method on
individual $S_n$-orbits, with suitable $\phi$ found for each
orbit. Thus while we cannot use the critical point algorithm of
\cite{faugere2020computing} directly we can make use of the various
subroutines used in it to construct a fast decision
procedure.
Intuitively, working with $S_n$-orbits is the same as separately
searching for real points having distinct coordinates, or real points
having two or more coordinates which are the same, or groups of
coordinates each of which has equal coordinates and so on.  In each
case an orbit can be described by points having $n$ or fewer pairwise
distinct coordinates, a key observation in constructing generic maps
invariant for each orbit. 

\begin{theorem}\label{thm:main}
 Let $\bm{f}=(f_1, \dots, f_s)$ be symmetric 
 polynomials in $\Q[x_1,
   \ldots, x_n]$ having maximal degree $d$. Assume that the Jacobian
 matrix of $\bm{f}$ with respect to $x_1, \dots, x_n$ has rank $s$ at any
 point of $V(\bm{f})$. Then there is a Monte Carlo algorithm ${\sf
   Real\_emptiness}$ which solves the real root decision problem for
 $\bm{f}$ with
\begin{multline*}
\softO\left(d^{6s+2}n^{11}{{n+d}\choose n}^6\left({{n+d}\choose n}
     +  {{n}\choose {s+1}} \right) \right) \\ \subset
 \left(d^s{{n+d}\choose n} {{n}\choose {s+1}} \right)^{O(1)}
\end{multline*}
operations in $\Q$.  Here the notion $\softO$ indicates that polylogarithmic factors are omitted. 
\end{theorem}
The remainder of the paper proceeds as follows. The next section
reviews known material, on invariant polynomials over products of
symmetric groups, the tools we use to work with $S_n$-orbits, and our
data structures. Section \ref{sec:smoothness} discusses our smoothness
requirement and shows that it is preserved by alternate
representations of invariant polynomials. Section \ref{sec:locus}
shows how we construct critical point functions along with their
critical point set. This is followed in Section \ref{sec:main} by a
description of our algorithm along with a proof of correctness and
complexity. The paper ends with a section on topics for future
research.



\section{Preliminaries}

\subsection{Invariant Polynomials}
\label{section:building_map}

We briefly review some properties of polynomials invariant under the
action of $S_{t_1}\times\cdots \times S_{t_k}$, with $S_{t_i}$ the
symmetric group on $t_i$ elements, for all $i$. In this paragraph, we work
with variables $\bmx = (\bmx_1, \ldots, \bmx_k)$, with each
$\bmx_i~=~(z_{1, i},\ldots, z_{t_i, i})$; for all $i$, the group
$S_{t_i}$ permutes the variables $\bmx_i$. For $j \ge 0$, we denote by
\[
E_{j,i} =  \sum_{1 \leq m_1 <m_2 < \cdots < m_j \leq t_i} z_{m_1,i} z_{m_2,i}
\cdots z_{m_j,i},
\]
the elementary  polynomial in the variables $\bm{z}_i$, with
each $E_{j, i}$ having degree~$j$, and by 
\[
\N_{j,i} = z_{1, i}^j + \cdots + z_{t_i, i}^j
\] 
the $j$-th Newton sum in the variables $\bmx_i$, for $i=1,\dots,k$.
The following two results are well-known.

For $i=1,\dots,k$, let $\bm{e}_i = (e_{1,i}, \dots, e_{t_i, i})$ be a set
of $t_i$ new variables and let $\bm{E}_i = (E_{1, i}, \dots, E_{t_i,
  i})$; we write $\bm{e}=(\bm{e}_1,\dots,\bm{e}_k)$ and $\bm{E}= (\bm{E}_1,\dots,\bm
E_k)$. 
\begin{lemma}\label{lemma:elemmap}
  Let $g\in [\bmx_1, \dots, \bmx_k]$ be invariant under the
  action of $S_{t_1}\times \cdots \times S_{t_k}$. Then there exists
  a unique $\gamma_g$ in $\field[\bm{e}]$ such that $g = \zeta_g(\bm{E})$.
\end{lemma}

Similarly, let $\n_{j,i}$ be new variables, and consider the sequences
$\bmn_i=(\n_{1, i}, \ldots, \n_{t_i, i})$ and $\bmn = (\bmn_1, \ldots,
\bmn_k)$, together with their polynomial counterparts $\bmN_i=(\N_{1,
  i}, \ldots, \N_{t_i, i})$ and $\bmN = (\bmN_1, \ldots, \bmN_k)$.

\begin{lemma}\label{lemma:newtonmap}
  Let $g\in [\bmx_1, \dots, \bmx_k]$ be invariant under the
  action of $S_{t_1}\times \cdots \times S_{t_k}$. Then there exists
  a unique $\zeta_g$ in $\field[\bmn]$ such that $g = \gamma_g(\Nmap)$.
\end{lemma}
\begin{example}
  Let 
  \[g = 2 (z_{1,1} z_{2,1} + z_{1,1}^2 + 2 z_{1,1} z_{2,1} + z_{2,1}^2)(z_{1,2}^2 + z_{2,2}^2),\]
  a polynomial invariant under $S_2 \times S_2$, with $\bmx_1 =
  (z_{1,1}, z_{2,1})$, $\bmx_2 = (z_{1,2}, z_{2,2})$, $k=2$ and $t_1=t_2=2$. In this
  case, we have
\[  g = (3 \N_{1,1}^2 - \N_{1,2})  \N_{2,2}\]
  and hence
  $
  \gamma_g =  (3 \n_{1,1}^2 - \n_{1,2})  \n_{2,2} \in \field[\n_{1,1},
    \n_{1,2}, \n_{2,1}, \n_{2,2}]. 
  $
\end{example}

\subsection{Describing $S_n$-orbits via Partitions}\label{sec:prelim}\label{sec:prelim:notation}

$S_n$-orbits are subsets of $\fieldC{}^n$ that play a central role in
our algorithm. In this section, we review notation and description of
$S_n$-orbits, along with the form of the output used in
\cite{faugere2020computing}.

A simple way to parameterize $S_n$-orbits is through the use of
partitions of $n$. A sequence $\lambda = (n_1^{t_1}\, \dots \,
n_k^{t_k})$, where $n_1 < \cdots < n_k$ and $n_i$'s and $t_i$'s are
positive integers, is called a partition of $n$ if $n_1t_1+\cdots +
n_kt_k = n$. The {\em length} of the partition $\lambda$ is defined as
$\ell := t_1 + \cdots + t_k$.

For a partition $\lambda = (n_1^{t_1}\, \dots \, n_k^{t_k})$ of $n$,
we use the notation from \cite[Section 2.3]{faugere2020computing} and
let $U_\lambda $ denote the set of all points $\bm{u}$ in $\fieldC{}^n$
that can be written as
\begin{multline} \label{eq:type}
 \bm{u} = (  \underbrace{u_{1, 1}, \dots, u_{1, 1}}_{n_1}, ~ 
    \dots,~ \underbrace{u_{t_1, 1}, \dots, u_{t_1,1}}_{n_1},~ 
 \dots,\\ ~ \underbrace{u_{1,k}, \dots, u_{1,k}}_{n_k},~
    \dots,~ \underbrace{u_{t_k, k}, \dots, u_{t_k, k}}_{n_k}).
\end{multline}
For any point $\bm{u}$ in $\fieldC{}^n$, we define
 its {\em type} as the unique partition $\lambda$ of $n$ such that
 there exists $\sigma \in S_n$ such that $\sigma(\bm{u}) \in U_{\lambda}$,
with the $u_{i,j}$'s in \eqref{eq:type} pairwise distinct.
Points of a given type $\lambda = (n_1^{t_1}\, \dots \, n_k^{t_k})$
are stabilized by the action of $S_\lambda := S_{t_1}\times \cdots
\times S_{t_k}$, the cartesian product of symmetric groups $S_{t_i}$.

For a partition $\lambda$ as above, we can then define a mapping $F_\lambda : U_\lambda \rightarrow
\fieldC{}^\ell$ as 
\begin{multline*}
\bm{u} {~\rm as~in~} \eqref{eq:type} \mapsto \\ (E_{1, i}(u_{1, i}, \dots,
u_{t_i, i}), \dots,E_{t_i, i}(u_{1, i}, \dots,u_{t_i, i}))_{1\le i \le k},
\end{multline*} where $E_{j, i}(u_{1, i}, \dots,u_{t_i, i})$  is the
$j$-th elementary 
symmetric function in $u_{1, i}, \dots,u_{t_i, i}$ for $i=1, \dots, k$
and $j=1, \dots, t_i$. One can think of the map $F_\lambda$ as a
compression of orbits. By applying this map, we can represent an
$S_n$-orbit $\mathcal{O}$ of type $\lambda$ by the single point
$F_\lambda(\mathcal{O}\cap U_\lambda).$

Furthermore, the map $F_\lambda$ is onto: for any $\bm{c} =  
(c_{1,1}, \dots,$ $c_{t_k, k}) \in \fieldC{}^\ell$, we define
polynomials $\rho_1(u), \dots, \rho_k(u)$ by 
\[
\rho_i(T) = T^{t_i} -c_{1, i}T^{t_i-1} + \cdots + (-1)^{t_i}c_{t_i, i}. 
\]
We can then find a point
$\bm{u} \in \fieldC{}^n$ in the preimage $F_\lambda^{~-1}(\bm{c})$ by
finding the roots $u_{1, i}, \dots, u_{t_i, i}$ of $\rho_i(T)$.

\subsection{Zero-Dimensional Parametrizations}
\label{subsec:para}
The subroutines we use from \cite{faugere2020computing} give their
output in terms of {\em zero-dimensional parametrizations}, which are
defined as follows. Let $W \subset \fieldC{}^n$ be a variety of
dimension zero, defined over $\field$. A zero-dimensional
parametrization $\scrR = ((v, v_1, \dots, v_n), \mu)$ of $W$ is 
\begin{itemize}
\item[(i)] a squarefree polynomial $v$ in $\field[t]$, where $t$ is a
  new indeterminate, and $\deg(v) = |W|$,
\item[(ii)] polynomials $v_1, \dots, v_n$ in $\field[t]$ such that
  $\deg(v_i) < \deg(v)$ for all $i$ and
\[
W = \left\{\left(\frac{v_1(\tau)}{v'(\tau)}, \dots,
    \frac{v_n(\tau)}{v'(\tau)}\right) \in \fieldC{}^n\, : \, v(\tau) = 0 \right\}, 
\]
\item[(iii)] a linear form $\mu$ in $n$ variables such that $\mu(v_1, \dots,
  v_n) = tv'$ (so the roots of $v$ are the values taken by $\mu$ on $W$).
\end{itemize}
When these conditions hold, we write $W = Z(\scrR)$. Representing the
points of $W$ by means of rational functions with $v'$ as denominator
is not necessary, but allows for a sharp control of the bit-size of
the output.

\section{Preserving smoothness}\label{sec:smoothness}

In our main algorithm, we assume that our input system
$\bm{f}=(f_1,\dots,f_s)$ satisfies the following smoothness condition
\begin{itemize}
\item[({\sf A})]: {\em the Jacobian matrix of $\bm{f}$ has rank $s$ at any point of $V(\bm{f})$.}
\end{itemize}
In this section, we discuss consequences of this assumption for
symmetric polynomials.

\smallskip\noindent{\bf Mapping to orbits: the map $\T_\lambda$.}  For
a partition $\lambda = (n_1^{t_1}\, \dots \, n_k^{t_k})$ of $n$, we
define the $\field$-algebra homomorphism $\T_\lambda:\field[x_1,
  \dots, x_n] \rightarrow \field[\bm{z}_1, \dots, \bm{z}_k]$, with
$\bmx_i=(z_{1, i},\ldots, z_{t_i, i})$ for all $i$, which maps the
variables $x_1, \dots, x_n$ to
\begin{multline}\label{eq:operator_T}
\underbrace{z_{1, 1}, \dots, z_{1, 1}}_{n_1}, ~
    \dots,~ \underbrace{z_{t_1, 1}, \dots, z_{t_1,1}}_{n_1},~
 \dots, \\ ~ \underbrace{z_{1,k}, \dots, z_{1,k}}_{n_k},~
    \dots,~ \underbrace{z_{t_k, k}, \dots, z_{t_k, k}}_{n_k}.
  \end{multline}

The operator $\T_\lambda$ extends to vectors of polynomials and
polynomial matrices entry-wise. The key observation here is that if
$f$ is symmetric, 
then its image through $\T_\lambda$ is
$S_{t_1}\times \cdots \times S_{t_k}$-invariant.

Fix a partition $\lambda = (n_1^{t_1}\, \dots \, n_k^{t_k})$ of $n$,
and let $\ell$ be its length.  Set
\[I_{j,i}:=\{\sigma_{j,i}+1,\dots,\sigma_{j,i}+n_i\}, 1 \le i \le k; 1\le j \le t_i\] with
$\sigma_{j,i}:=\sum_{r=1}^{i-1}t_{r}n_{r}+(j-1)n_i$. Variables $x_m$, for $m$ in
$I_{j,i}$, are precisely those that map to $z_{j,i}$ under $\T_\lambda$. Define
further the matrix $\bm{Z} \in \Q^{\ell \times n}$ with $\ell = t_1 +\cdots +
t_k$, where rows are indexed by pairs $(j,i)$ as above and columns by $m \in
\{1,\dots,n\}$. For all such $(j,i)$, the entry of row index $(j,i)$ and column
index $m \in I_{j,i}$ is set to $1/n_i$, all others are zero.  In
other words, $\bm{Z}  = {\rm diag}(\bm{Z}_1, \dots, \bm{Z}_k)$, where 
\[
\bm{Z}_i = \left( 
\begin{matrix}
\begin{matrix}\frac{1}{n_i} & \cdots & \frac{1}{n_i} \end{matrix} & {\bf
    0} & \cdots & {\bf 0}\\
{\bf 0} &\begin{matrix}\frac{1}{n_i} & \cdots &
  \frac{1}{n_i} \end{matrix} & \cdots & {\bf 0} \\
\vdots &  &\ddots & \vdots \\
{\bf 0} & {\bf 0} & \cdots & \begin{matrix}\frac{1}{n_i} & \cdots &
  \frac{1}{n_i} \end{matrix} 
\end{matrix} 
\right) 
\] is a matrix in $\Q^{t_i \times n_it_i}$.

\begin{example}\label{ex:matZ}
Consider the partition $\lambda = (2^2\, 3^1)$ of $n=7$. Then
$n_1=2, t_1=2$, $n_2=3$, $t_2=1$ and the length of $\lambda$ is
$3$. In this case,
\[
\bm{Z} =  
\left( \begin{matrix} \frac{1}{2} & \frac{1}{2} & \\
 &  & \frac{1}{2} & \frac{1}{2} \\
 &  & &  & \frac{1}{3}& \frac{1}{3}& \frac{1}{3}
\end{matrix}\right).
\]
\end{example}

\begin{lemma}
  Let $\bm{f}= (f_1, \dots, f_s) \subset \field[x_1, \dots, x_n]$ be a sequence of symmetric 
  polynomials, and let $\lambda$ be a partition of $n$. Then
  \[
 \T_\lambda(\jac_{x_1, \dots, x_n}(\bm{f}))= \jac_{\bm{z}_1, \dots,
   \bm{z}_k}(\T_\lambda(\bm{f})) \cdot \bm{Z},\] where $\bm{Z}$ is the matrix
 defined above.\label{lemma:decom}
\end{lemma}

\begin{proof} For any polynomial $f$ in $\field[x_1, \dots,
    x_n]$, applying the operator $\T_\lambda$ on $f$ evaluates $f$ at
  $x_m = z_{j,i}$ for $1 \le i \le k$, $1\le j \le t_i$ and $m$ in
  $I_{j,i}$. By the multivariable chain rule,
\begin{align*}
\frac{\partial \T_\lambda(f)}{\partial z_{j,i}} &= \sum_{m \in
  I_{j,i}} \T_\lambda\left( \frac{\partial f}{\partial
  x_m}\right ). 
\end{align*}
If $f$ is symmetric, 
for $m,m'$
in $I_{j,i}$, we then  have
\[\T_\lambda\left( \frac{\partial f}{\partial x_m}\right ) =
\T_\lambda\left( \frac{\partial f}{\partial 
  x_{m'}}\right ),\]
so that, for $m$ in $I_{j,i}$,
\[ \T_\lambda\left( \frac{\partial f}{\partial x_m}\right ) = \frac
1{n_i} \frac{\partial \T_\lambda(f)}{\partial z_{j,i}}.\] 
This argument can be extended to a sequence of polynomials to obtain
our claim.
\end{proof}

\begin{example}
We continue  Example \ref{ex:matZ} with a single $S_7$-invariant
polynomial $f= \sum_{1 \le i \le j \le   7}x_ix_j$. Then $$\T_\lambda(f) = 3z_{1,1}^2+ 3z_{2,1}^2 +
6z_{1,2}^2+6z_{1,1}z_{1,2}+4z_{1,1}z_{2,1}+6z_{1,2}z_{2,1},$$ and so
\[\jac(\T_\lambda(f)) = (6 z_{1,1} + 6 z_{1,2} + 4 z_{2,1}, 4 z_{1,1} +
6 z_{1,2} + 6 z_{2,1}, 6 z_{1,1} + 12 z_{1,2} + 6 z_{2,1}).\] This
implies that $\jac(\T_\lambda(f)) \cdot \bm{Z}$ is equal to $(u, u, v, v, w, w, w)$,
with
\[u=3 z_{1,1} + 3 z_{1,2} + 2 z_{2,1}, v=2 z_{1,1} + 3 z_{1,2} + 3 z_{2,1}, w=2 z_{1,1} + 4 z_{1,2} + 2 z_{2,1}.\]
This is precisely $\T_\lambda(\jac(f))$. 
 \end{example} 

\begin{corollary}\label{lemma:extend_slice}
  Under the assumptions of the previous lemma, if $\bm{f}$ satisfies
  condition $({\sf A})$, then $\T_\lambda(\bm{f}) \subset \field[\bm{z}_1, \dots,
    \bm{z}_k]$ does as well.
\end{corollary}

\begin{proof}
  Let $\bm{\alpha} = (\alpha_{1,1}, \dots, \alpha_{t_1, 1}, \dots,
  \alpha_{1, k}, \dots, \alpha_{t_k k})$ be a zero of $\T_\lambda(\bm{f})$
  in $\fieldC^\ell$. We have to prove that $\jac_{\bm{z}_1, \dots,
    \bm{z}_k}(\T_\lambda(\bm{f}))(\bm{\alpha})$ has a trivial left kernel.
  
  Consider the point
  \begin{multline}\bm{\varepsilon} = \big(\underbrace{\alpha_{1, 1},
      \dots, \alpha_{1, 1}}_{n_1}, \dots,  \underbrace{\alpha_{t_1, 1},
      \dots, \alpha_{t_1,1}}_{n_1}, \dots, \\  \underbrace{\alpha_{1,k},
      \dots, \alpha_{1,k}}_{n_k}, \dots, \underbrace{\alpha_{t_k, k}, \dots,
      \alpha_{t_k, k}}_{n_k} \big) \in \fieldC^n,\label{eq:back}
  \end{multline} which lies in $V(\bm{f})$.
  In particular, for any  $g$ in $\Q[x_1,\dots,x_n]$, we
  have $\T_\lambda(g)(\bm\alpha) = g(\bm
  \varepsilon)$. 
  Applying this to the Jacobian matrix of $\bm{f}$, we obtain
  $\T_\lambda(\jac(\bm{f}))(\bm\alpha) = \jac(\bm{f})(\bm\varepsilon).$
  Since by assumption $\bm{f}$ is symmetric, 
  the previous lemma implies that
  \[\jac(\bm{f})(\bm\varepsilon) = \jac_{\bm{z}_1, \dots, \bm{z}_k}(\T_\lambda(\bm{f}))(\bm{\alpha}) \cdot \bm
    Z.\] Since $ \jac(\bm{f})(\bm\varepsilon)$ has rank $s$ (by condition
    ${\sf A}$), the left kernel of $ \jac(\bm{f})(\bm\varepsilon)$ is
    trivial.

    It follows that the left kernel of $\jac_{\bm{z}_1, \dots,
      \bm{z}_k}(\T_\lambda(\bm{f}))(\bm{\alpha})$ is also trivial.
\end{proof}

When we represent $S_{t_1}\times \cdots \times S_{t_k}$-invariant
functions in terms of Newton sums, we can show that the new
representation also preserves condition ({\sf A}).

\begin{lemma}\label{lemma:transferreg}
  Assume $(g_1, \ldots, g_s)\subset \field[\bm{z}_1, \dots, \bm{z}_k]$ is
  $S_{t_1}\times \cdots \times S_{t_k}$- invariant and satisfies
  condition $({\sf A})$.  If we set $h_i = \gamma_{g_i}$ for all $i$,
  then $(h_1, \ldots, h_s)$ also satisfies condition $({\sf A})$.
\end{lemma}

\begin{proof}
  The Jacobian matrix $\jac(\bm{g})$ of $(g_1, \ldots, g_s)$ factors as
  \[
  \jac(\bm{g}) = \jac(\h)(\Nmap) \cdot \bmMV, \ {\rm where}
  \ \bmMV = {\rm diag}(V_1, \dots, V_k) 
  \]
  with each $V_i$  a row-scaled Vandermonde matrix given by
  \begin{equation} \label{eq:vander}
  V_i = {\small
  \left( \begin{matrix} 1& & &  \\
     & 2 & & \\
      & & \ddots &  \\
     &  &  & {t_i}
  \end{matrix}\right)
  \left( \begin{matrix} 1& 1& \cdots & 1 \\ z_{1, i} & z_{2, i}
    &\cdots & z_{t_i, i} \\ \vdots & & & \vdots \\ z_{1, i}^{t_i-1} &
    z_{2, i}^{t_i-1} &\cdots & z_{t_i, i}^{t_i-1}
  \end{matrix}\right) }.
  \end{equation}
  Let $\bm{\eta} $ be a point in the vanishing set of $(h_1, \ldots,
  h_s)$ and let $\bm{\varepsilon}$ be in $\Nmap^{-1}(\bm{\eta})$. If
  $\jac(\h)$ is rank deficient at $\bm{\eta}$ then $\jac(\h)(\Nmap)(\bm
  \varepsilon)$ is also rank deficient. This implies that the rank of
  $\jac(\bm{g})(\bm{\varepsilon})$, which is bounded above by those of
  $\jac(\h)(\Nmap)(\bm{\varepsilon})$ and $\bmMV(\bm{\varepsilon})$, is
  deficient.
\end{proof}

Similarly, instead of using a row-scaled Vandermonde matrix $V_i$ as
in \eqref{eq:vander}, we can use $V_i$ as the Jacobian matrix of
elementary symmetric functions in $\bm{z}_i$. This gives a similar result
but for the polynomials $\zeta_{g_1}, \dots, \zeta_{g_s}$. 

\begin{lemma}\label{lemma:transferreg_ele}
  Assume $(g_1, \ldots, g_s)\subset \field[\bm{z}_1, \dots, \bm{z}_k]$ is
  $S_{t_1}\times \cdots \times S_{t_k}$- invariant and satisfies
  condition $({\sf A})$. Then the sequence of polynomials
  $(\zeta_{g_1}, \ldots, \zeta_{g_s})$ also satisfies condition $({\sf
    A})$.
\end{lemma}

\section{Critical loci}\label{sec:locus}

If $W \subset \C^\ell$ is an equidimensional algebraic set, and $\phi$
a polynomial function defined on $W$, a non-singular point $\bm{w} \in W$
is called a {\em critical point} of $\phi$ on $W$ if the gradient of
$\phi$ at $\bm{w}$ is normal to the tangent space $T_{\bm{w}} W$ of $W$ at
$\bm{w}$.

If $\bm{g}=(g_1,\dots,g_s)$ are generators of the ideal associated to $W$, then
$T_{\bm{w}}W$ is the right kernel of the Jacobian matrix $\jac(\bm{g})$ of $\bm{g}$
evaluated at $\bm{w}$. In the cases we will consider, this matrix will have rank
$s$ at all points of $W$ (that is, $\bm{g}$ satisfies condition ${\sf A}$). The set
of critical points of the restriction of $\phi$ to $W$ is then defined by the
vanishing of $\bm{g}$, and of the $(s+1)$-minors of the Jacobian matrix $\jac({\bm{g},
  \phi})$ of $\bm{g}$ and $\phi$.

\subsection{Finiteness through genericity}\label{ssec:finitestatement}

Let $\bm{g} = (g_1, \ldots, g_s)$ in $\K[\bm{z}_1, \dots. \bm{z}_k]$ with each
$g_i$ invariant under the action of $S_{t_1}\times\cdots \times
S_{t_k}$; we write $\ell = t_1+\cdots +t_k$.  We introduce some useful
$S_{t_1}\times\cdots \times S_{t_k}$-invariant mappings and discuss
the properties of their critical points on $V(\bm{g}) \subset
\fieldC{}^\ell$.

For $1\leq i \leq k$, let $\fraka_i = (\fraka_{1, i}, \ldots,
\fraka_{\macrodelta_i, i})$ be new indeterminates, and recall that $P_{j,
  i}$ is the $j$-th Newton sum for the variables $\bm{z}_i$. Set
\begin{equation}
\label{eq:calNN}
\calN_{\fraka} =
\sum_{i=1}^k c_i \N_{t_i+1, i} +
\sum_{i=1}^k\sum_{j=1}^{\macrodelta_i}\fraka_{j,i} \N_{j,i}
\end{equation} where $c_i=1$ if
$\macrodelta_i$ is odd and $c_i=0$ if $\macrodelta_i$ is even.  So
$\calN_{\fraka}$ has even degree and is invariant under the action
of $S_{t_1}\times \cdots \times S_{t_k}$. 
For $\a=(\a_1,\dots,\a_k)$ in $\C^{\macrodelta_1} \times \cdots \times
\C^{\macrodelta_k}$, with each $\a_i$ in $\C^{\macrodelta_i}$, we
denote by $\calN_{\a}$ the polynomials in $\C[\bm{z}_1, \dots, \bm{z}_k]$
obtained by evaluating the indeterminates $\fraka_i$ at $\a_i$ in $\calN_{\fraka}$, for all
$i$.

Further, we denote by $\calU \subset \Kbar{}^\ell$ the open set consisting of
points $\bm{w} = (\bm{w}_1, \dots, \bm{w}_k)$ such that the coordinates of $\bm
w_i$ are pairwise distinct for $i=1, \dots, k$. Note that $\calU$ depends on the
partition $\lambda = (n_1^{t_1} \dots n_k^{t_k})$; when needed because of the
use of different partitions, we will denote it by $\calU_{\lambda}$.

\begin{proposition}\label{prop:feasible}
  Let $\bm{g} = (g_1, \dots, g_s)$ be $S_{t_1}\times\cdots \times S_{t_k}$-invariant
  polynomials in $\K[\bm{z}_1, \dots, \bm{z}_k]$. Suppose further that $\bm{g}$ satisfies
  condition {\sf (A)}. Then there exists a non-empty Zariski open set
  $\calA\subset \Kbar{}^{\macrodelta_1}\times \cdots \times
  \Kbar{}^{\macrodelta_k}$ such that for $\a\in \calA$, the restriction of
  $\calN_{\a}$ to $V(\bm{g})$ has finitely many critical points in $\calU$.
\end{proposition}

\subsection{Proof of Proposition \ref{prop:feasible}}
For new variables $\vl_1, \ldots, \vl_s$, we denote by
$\calS_{\fraka}$ the polynomials
\[
\calS_{\fraka}  = \big(g_1, \ldots, g_s, \quad  [\vl_1 \cdots \vl_s~1] \cdot \jac(\bm{g}, \calN_{\fraka}) \big).
\] For $\a=(\a_1,\dots,\a_k)$ in $\C^{\macrodelta_1} \times \cdots \times \C^{\macrodelta_k}$, with each $\a_i$
in $\C^{\macrodelta_i}$, we denote by $\calS_{\a}$ the polynomials in
$\C[\vl_1,\dots,\vl_s,\bm{z}_1, \dots. \bm{z}_k]$ obtained by evaluating
$\fraka_i$ at $\a_i$ in $\calS_{\fraka}$, for all $i$. Finally,
denote by $\pi$ the projection from the $(L,\bm{z})$-space $\C^{s + \ell}$ to the $\bm{z}$-space $\C^{\ell}$. 
\begin{lemma}\label{lemma:proj}
  Suppose that $\bm{g}$ satisfies condition {\sf (A)}.  Then for $\a\in
  \Kbar{}^{\macrodelta_1}\times \cdots \times \Kbar{}^{\macrodelta_k}$,
  $\pi(V(\calS_{\a}))$ is the critical locus of the restriction of the
  map $\calN_{\a}$ to $V(\bm{g})$.
\end{lemma}

\begin{proof}
  For any $\a\in \Kbar{}^{\macrodelta_1}\times \cdots \times
  \Kbar{}^{\macrodelta_k}$, we denote by $W(\calN_{\a}, \bm{g})$ the set of critical
  points of the restriction of $\calN_{\a}$ to $V(\bm{g})$. Since $\bm{g}$ satisfies
  condition {\sf (A)}, the set $W(\calN_{\a}, \bm{g})$ is given by
  \[
 \{ \bm{w}  ~ |  g_1(\bm{w}) = \cdots = g_s(\bm{w}) = 0, \quad \rank(\jac(\bm{g},\calN_{\a})(\bm{w})) \le s \}.
  \]
  Consider $\bm{w}$ in $W(\calN_{\a}, \bm{g})$ and a nonzero vector $\bm{c}$ 
  in the left kernel of $\jac(\bm{g},\calN_{\a})(\bm{w})$, of the form $\bm
  \c=(c_1,\dots,c_s,c_{s+1})$. The last coordinate $c_{s+1}$ cannot
  vanish, as otherwise $(c_1,\dots,c_s)$ would be a nonzero vector in
  the left kernel of $\jac(\bm{g})(\bm{w})$ (which is ruled out by
  condition {\sf (A)}).  Dividing through by $c_{s+1}$, the point
  $(\bm{c}',\bm{w})$, with $c'_i=c_i/c_{s+1}$ for $i=1,\dots,s$, is a
  solution of $\calS_{\a}$.

  Conversely, take $( \bm{\ell},\bm{w})$ in $V(\calS_{\a})$. Thus, $\bm{w}$
  cancels $\bm{g}$, and $\jac(\bm{g}, \calN_{\a})$ has rank less than $s+1$ at
  $\bm{w}$, so that $\pi(V(\calS_a))$ is in $W(\calN_{\a}, \bm{g})$.
\end{proof}
Let $\calN_{\fraka}$ and $\gamma_{\calN_{\fraka}}$ be defined as in
\eqref{eq:calNN} and Lemma~\ref{lemma:newtonmap}, respectively.  For
$i=1,\dots,k$, set $Q_{i} = \gamma_{\N_{t_i+1,i}}$, and
let $h_1,\dots,h_s =\gamma_{g_1},\dots,\gamma_{g_s}$.
In particular, Lemma
\ref{lemma:newtonmap} implies that $\gamma_{\calN_{\fraka}}$ is given
by
\[
\sum_{i=1}^k c_i Q_{i} +
\sum_{i=1}^k\sum_{j=1}^{\macrodelta_i}\fraka_{j,i}\n_{j,i}.
\] 
The sequence $\calS_{\fraka}$ can be rewritten as
\begin{multline*}\hspace{-0.3cm} h_1\circ \bmN, \ldots, h_s\circ \bmN, \\ 
[\vl_1~\ldots~\vl_s~1]
  \left(\begin{matrix}
    \frac{\partial h_1}{\partial \n_{1,1}} & \cdots & \frac{\partial
                                                     h_1}{\partial
                                                     \n_{t_k, k}}
                                                   \\ 
    \vdots &  & \vdots \\
    \frac{\partial h_s}{\partial \n_{1, k}} & \cdots & \frac{\partial
                                                      h_s}{\partial \n_{t_k,
                                                      k}} 
                                                  \\ 
{c_1 \frac{\partial
        Q_{1}}{\partial \n_{1,1}}+}    \fraka_{1,1} & \cdots & {c_k\frac{\partial
        Q_{k}}{\partial \n_{t_k, k}} + }\fraka_{t_k, k} 
  \end{matrix}\right)_{\bmN(\bm{z})} \cdot \bmMV,
\end{multline*}
where $\bmMV$ is a multi-row-scaled Vandermonde matrix which is the Jacobian
matrix of $\bmN$ with respect to $\bm{z}$. This matrix has full rank at any point
$\bm{w}$ in the open set $\calU$ defined in
Subsection~\ref{ssec:finitestatement}.

In particular, for any $\a\in
\Kbar{}^{\macrodelta_1}\times\cdots\times \Kbar{}^{\macrodelta_k}$,
the intersection of $V(\calS_{\a})$ with $\C^s\times \calU$ is contained
in the preimage by the map ${\rm Id}\times \bmN$ of the vanishing set
of the sequence
\begin{multline*}\hspace{-0.3cm}\bmH_{\a} : \quad  h_1, \ldots, h_s, \\
[\vl_1~\cdots~\vl_s~1]
  \left(\begin{matrix}
    \frac{\partial h_1}{\partial \n_{1,1}} & \cdots & \frac{\partial
                                                     h_1}{\partial \n_{t_k, k}} \\
    \vdots &  & \vdots \\
    \frac{\partial h_s}{\partial \n_{1, 1}} & \cdots & \frac{\partial
                                                      h_s}{\partial \n_{t_k, k}} \\
    {c_1\frac{\partial
        Q_{1}}{\partial \n_{1,1}}+} a_{1,1}  & \cdots & {c_k\frac{\partial
        Q_{k}}{\partial \n_{t_k, k}}+ } a_{t_k, k}\\
  \end{matrix}\right).
\end{multline*}
Since for all $1\leq i \leq k$, $\bmN_i$ defines a map with finite fibers (by
Newton identities and Vieta's formula, the preimage by $\bmN$ of some point is
the set of roots of some polynomial of degree $t_i$), we deduce that $\bmN$ and
consequently $ {\rm Id}\times \bmN$ define maps with finite fibers. Thus 
\begin{lemma}\label{lemma:HtoS}
  If $V(\bmH_{\a})$ is finite, then $V(\calS_{\a})\cap (\C^{s}\times \calU)$ is
  finite.
\end{lemma}
It remains to investigate finiteness properties of $V(\bmH_{\a})$. 
\begin{proposition} \label{prop:newpols}
  Suppose that $\h$ satisfies condition {\sf (A)}. Then, there exists
  a non-empty Zariski open set $\calA\subset \Kbar{}^{\macrodelta_1}\times
  \cdots \times \Kbar{}^{\macrodelta_k}$ such that for any $\a\in \calA$,
  $\langle \bmH_{\a} \rangle \subset \Kbar[\vl_1,\dots,\vl_s,\bm{z}_1,
    \dots, \bm{z}_k]$ is a radical ideal whose zero-set is finite.
\end{proposition}

\begin{proof}
  Let $W\subset \Kbar{}^{t_1}\times \cdots\times \Kbar{}^{t_k}$ be the
  vanishing set of $(h_1, \ldots, h_s)$. Consider now the map
\begin{multline*}
\hspace{-.3cm}(\bm{\eta}, \bm{w})\in \Kbar{}^{s}\times W
 \to \\- \left( \sum_{i=1}^s\bm\eta_i\frac{\partial
  h_i}{\partial \n_{1,1}}{+c_1\frac{\partial Q_{1}}{\partial
    \n_{1,1}}} \right)_{(\bm{w})}, \ldots, -\left(
\sum_{i=1}^s\bm\eta_i\frac{\partial h_i}{\partial \n_{t_k, k}}
    {+c_k\frac{\partial Q_{k}}{\partial \n_{t_k,k}}}\right)_{(\bm{w})}.
\end{multline*}
  By Sard's theorem \cite[Chap. 2, Sec. 6.2, Thm
    2]{shafarevich1994basic}, the set of critical values of this map
  is contained in a proper Zariski closed set $\mathcal{B}$ of
  $\Kbar{}^{t_1}\times \cdots \times \Kbar{}^{t_k}$. Since $\h$
  satisfies condition ${\sf (A)}$, for $\a$ outside $\mathcal{B}$, the
  Jacobian matrix of $\bmH_{\a}$ has full rank at any $(\bm{\eta}, \bm
  w)$ with $\bm{w}$ in $W$.  Hence, by the Jacobian criterion \cite[Thm
    16.19]{eisenbud2013commutative}, the ideal generated by $\bmH_{\a}$
  in $\Kbar[\vl_1,\dots,\vl_s,\bm{z}_1, \dots, \bm{z}_k]$ is radical and is of
  dimension at most zero.
\end{proof}

\begin{proof}[Proof of Prop \ref{prop:feasible}]
  Let $\calA$ be the non-empty Zariski open set defined in Prop
  \ref{prop:newpols}.  Since $\bm{g}$ satisfies condition $({\sf A})$, Lemma
  \ref{lemma:proj} implies that, for any $\a \in \calA$, the critical
  locus of the map $\calN_{\a}$ restricted to $V(\bm{g})$ is equal to
  $\pi(V(\calS_{\a}))$. In addition, the sequence $(\h)$ also satisfies
  condition $({\sf A})$ by Lemma \ref{lemma:transferreg}. Then, by
  Prop.~\ref{prop:newpols}, for any $\a \in \calA$, the algebraic set
  defined by $\bmH_{\a}$ is finite. 

  By Lemma~\ref{lemma:HtoS}, this implies that $V(\calS_{\a})$ contains
  finitely many points in $\C^s\times\mathcal{U}$. This finishes
  our proof of Prop.~\ref{prop:feasible}.
\end{proof}

Using techniques from~\cite{ElGiSc20}, one could give a simple
exponential upper bound the degree of a hypersurface containing the
complement of $\calA$.

\subsection{Finding extrema using  proper maps}
\label{sec:proper}

A real valued function $\psi : \R^n \rightarrow \R$ is {\em proper} at
$x \in \R$ if there exists an $\varepsilon>0$ such that 
$\psi^{-1}([x-\varepsilon,x+\varepsilon])$ is compact. Such functions
are of interest because a proper polynomial restricted to a real
algebraic set $W$ reaches extrema on each connected component of $W$. Using
\cite[Thm~2.1 and Cor 2.2]{sakkalis2005note} one can construct proper
polynomials in the following way.

Let $F = F_k(x_1, \dots, x_n) + F_{k-1}(x_1, \dots, x_n) + \cdots +
F_0(x_1, \dots, x_n): \R^n \rightarrow \R$ be a real polynomial, where
$F_i$ is the homogeneous component of degree $i$ of $F$. Assume
further that the leading form $F_k$ of $F$ is positive definite; then,
$F$ is proper. In particular, the map $P_{2m} + \sum_{i=0}^{2m-1}
\lambda_i P_{i}$, with $P_i$ the Newton sums in $x_1,\dots,x_n$ and
all $\lambda_i$ in $\K$, is proper. We can extend this to blocks of
variables.

\begin{lemma} Let $\bm{z}_1, \dots, \bm{z}_k$ be blocks of $t_1, \dots,
  t_k$ variables, respectively. If $P_{j,i} := z_{1,i}^j + \cdots +
  z_{t_{i},i}^j$, then for any $m_1, \dots, m_k \ge 1$ and coefficients
  $\lambda_{i,j}$ in $\K$, the
  map
  \[\sum_{i=1}^k P_{2m_i,i} + \sum_{i=1}^k \sum_{j=0}^{2m_i-1} \lambda_{j,i} P_{j,i}\]
  is
  proper.\label{lemma:proper_power}
\end{lemma}

\section{Main result}\label{sec:main}

Let $\bm{f} = (f_1, \ldots, f_s)$ be a sequence of symmetric polynomials
in $\K[x_1, \dots, x_n]$ that satisfies condition {\sf (A)}. In this section we present an algorithm and its complexity to decide
whether the real locus of $V(\bm{f})$ is empty or not.

To exploit the symmetry of $\bm{f}$ and to decide whether the set
$V_{\R}(\bm{f})$ is empty or not, our main idea is slicing the variety
$V(\bm{f})$ with hyperplanes which are encoded by a partition $\lambda$ of
$n$. This way, we obtain a new polynomial system which is invariant
under the action $S_{\lambda} := S_{t_1} \times \cdots \times S_{t_k}$
of symmetric groups. We proved in Lemma \ref{lemma:extend_slice} that
this new system also satisfies condition {\sf (A)}. We then use the
critical point method to decide whether the real locus of the
algebraic variety defined by this new system is empty or not by taking
a $S_{\lambda}$-invariant map as defined in the previous section.

\subsection{Critical points along $S_n$-orbits}

Let $\bm{g} = (g_1, \dots, g_s)$ be a sequence of $S_{\lambda}$-invariant
polynomials and $\phi$ be a $S_\lambda$-invariant map in $\K[\bm{z}_1,
  \dots, \bm{z}_k]$, with $\bm{z}_i=(z_{1, i}, \dots, z_{t_i, i})$ for all
$i$. As before, we set $\ell = t_1 + \cdots + t_k$, and we assume that
$s \le \ell$. Assume further that the sequence $\bm{g}$ satisfies
condition $({\sf A})$. Let $\phi$ be a $S_\lambda$-invariant map in $\K[\bm{z}_1,
  \dots, \bm{z}_k]$. 

Let $\zeta_{\phi}$ and $\zeta_{\bm{g}}$ in $\K[\bm{e}_1, \dots, \bm{e}_k]$, where
$\bm{e}_i = (e_{1,i}, \dots, e_{t_i, i})$ is a set of $t_i$ new
variables, be such that  
\[
 \phi = \zeta_{\phi}(\bm{E}_1, \dots, \bm{E}_k) \quad {\rm and} \quad \bm{g} =
 \zeta_{\bm{g}}(\bm{E}_1, \dots, \bm{E}_k). 
\]  
Here $\bm{E}_i = (E_{1, i}, \dots, E_{t_i, i})$ denotes the
vector of elementary symmetric polynomials in variables $\bm{z}_i$, with
each $E_{j, i}$ having degree~$j$ for all $j, i$.

\begin{lemma} \label{lemma:subroutine}
  Let $\bm{g}, \phi$, and $\lambda$ as above. Assume further that
  $\zeta_\phi$ has finitely many critical points on
  $V(\zeta_{\bm{g}})$. Then there exists a randomized algorithm ${\sf
    Critical\_points}$ $(\bm{g}, \phi, \lambda)$ which returns a
  zero-dimensional parametrization of   the critical points of
  $\zeta_{\phi}$ restricted to $V(\zeta_{\bm{g}})$. 
  The algorithm uses
  $$\softO\left(\delta^2c_\lambda(e_\lambda + c_\lambda^5)n^4\Gamma\right)$$
  operations in $\K$, where
  \begin{align*}&c_\lambda  = \frac{\deg(g_1)\cdots \deg(g_s)
      \cdot E_{\ell-s}(\delta-1, \dots, \delta-\ell)}{t_1!\cdots t_k!}, \\  &\Gamma = 
    n^2{{n+\delta}\choose \delta} + n^4 {{n}\choose {s+1}}, \ {\rm and}\\
    &e_\lambda = \frac{n(\deg(g_1)+1)\cdots (\deg(g_s)+1)\cdot
      E_{\ell-s}(\delta, \dots,  \delta-\ell+1)}{t_1!\cdots t_k!},
  \end{align*} with  $\delta = \max(\deg(\bm{g}), \deg(\phi))$.
  The number of solutions is at most $c_\lambda$.
\end{lemma}

\begin{proof}
  The ${\sf Critical\_points}$ procedure contains two steps: first
  finding $\zeta_{\bm{g}}$ and $\zeta_{\phi}$ from $\bm{g}$ and $\phi$ and
  then computing a representation for the set $W(\zeta_{\phi},
  \zeta_{\bm{g}})$ of critical points of $\zeta_{\phi}$ on
  $V(\zeta_{\bm{g}})$. The first step can be done using the algorithm
  ${\sf Symmetric\_Coordinates}$ from \cite[Lemma
    9]{faugere2020computing}, which uses $\softO\left({{\ell+\delta}\choose
    \delta}^2\right)$ operations in $\K$.

  Since the sequence $\bm{g}$ satisfies condition $({\sf A})$,
  Lemma~\ref{lemma:transferreg_ele} implies that $\zeta_{\bm{g}}$ also
  satisfies condition $({\sf A})$. Then, the set $W(\zeta_{\phi},
  \zeta_{\bm{g}})$ is the zero set of $\zeta_{\bm{g}}$ and all the
  $(s+1)$-minors of $\jac(\zeta_{\bm{g}}, \zeta_{\phi})$. In particular,
  when $\ell = s$, $W(\zeta_{\phi}, \zeta_{\bm{g}})=V(\zeta_{\bm{g}})$.

  Since each $E_{j, i}$ has degree $j$, it is natural to assign a
  weight $j$ to the variable $e_{j, i}$, so that the polynomial ring
  $\K[\bm{e}_1, \dots, \bm{e}_k]$ is weighted of weights $(1, \dots, t_1,
  \dots, 1, \dots, t_k)$. The weighted degrees of $\zeta_{\bm{g}}$ and
  $\zeta_{\phi}$ are then equal to those of $\bm{g}$ and $\phi$,
  respectively.  To compute a zero-dimensional parametrization for
  $W(\zeta_{\phi}, \zeta_{\bm{g}})$ we use the symbolic homotopy method for
  weighted domain given in \cite[Thm 5.3]{labahn2021homotopy} (see
  also \cite[Sec 5.2]{faugere2020computing} for a detailed complexity
  analysis). This procedure is randomized and requires
\[
\softO\left(\delta^2c_\lambda(e_\lambda + c_\lambda^5)n^4\Gamma\right) \
{\rm operations \ in \   \K}.
\] Furthermore, results from \cite[Thm 5.3]{labahn2021homotopy} also
imply that the number of points in the output is at most $c_\lambda$. 

Thus, the total complexity of the ${\sf Critical\_points}$ algorithm 
is then  $\softO\left( \delta^2c_\lambda(e_\lambda +
  c_\lambda^5)n^4\Gamma\right)$ operations in $\K$.  
\end{proof}

\subsection{The {\sf Decide} procedure}

Consider a partition $\lambda = (n_1^{t_1}\, \dots \, n_k^{t_k})$ of
$n$, and let \[\scrR_\lambda = (v, v_{1,1}, \dots, v_{t_1, 1}, \dots,
v_{1,k},\dots, v_{t_k, k},\mu)\] be a parametrization which encodes a
finite set $W_\lambda\subset \C^\ell $. This set lies in the target space
of the algebraic map $F_\lambda : U_\lambda \rightarrow \Kbar{}^\ell$
defined in Subsection~\ref{sec:prelim:notation} as
\begin{multline}
   \bm{u} = ( \underbrace{u_{1, 1}, \dots, u_{1, 1}}_{n_1}, ~ 
  \dots,~ \underbrace{u_{t_k, k}, \dots, u_{t_k, k}}_{n_k}) \\
  \mapsto (E_{1, i}(u_{1, i}, \dots,
  u_{t_i, i}), \dots,E_{t_i, i}(u_{1, i}, \dots,u_{t_i, i}))_{1\le i \le k},
\end{multline} where $E_{j, i}(u_{1, i}, \dots,u_{t_i, i})$  is the
$j$-th elementary 
symmetric function in $u_{1, i}, \dots,u_{t_i, i}$ for $i=1, \dots, k$
and $j=1, \dots, t_i$.
 
Let $V_\lambda$ be the preimage of $W_\lambda$  by $F_\lambda$. In
this subsection we present a procedure called ${\sf
  Decide}(\scrR_\lambda)$ which takes as input $\scrR_\lambda$, and
decides whether the set $V_\lambda$ contains real points. 

In order to do this, a straightforward strategy consists in solving
the polynomial system to invert the map $F_\lambda$. Because of the
group action of $S_{t_1}\times \cdots \times S_{t_k}$, we would then
obtain $t_1 !\cdots t_k ! $ points in the preimage of a single point
in $W_\lambda$: we would lose the benefit of all that had been done
before.

This difficulty can be bypassed by encoding one single point per orbit
in the preimage of the points in $W_\lambda$. This can be done via
the following steps.
\begin{itemize}
\item[(i)] Group together the variables $\bm{e}_i = (e_{1,i}, \ldots, e_{t_i,i})$ which 
  encode the values taken by the elementary symmetric functions $E_{i,1},
  \ldots, E_{i, t_i}$ (see Sec. \ref{sec:prelim:notation}) and denote by
  $v_{i,1},\ldots, v_{i, t_i}$ the parametrizations corresponding to
  $e_{1, i},  \ldots, e_{t_i, i}$;
\item[(ii)]  Make a reduction to a bivariate polynomial system by
  considering the polynomial with coefficients in $\K[t]$
  \[
    \rho_i = v'u^{t_i} - v_{1, i} u^{t_i - 1} + \cdots + (-1)^{t_i}
    v_{t_i, i} \in \K[t][u]
  \]
  and ``solving'' the system $\rho_i = v = 0$. Here we recall that $v\in \K[t]$ and 
  is square-free,  so that $v$ and $v'$ are coprime.
\item[(iii)]  It remains to decide whether, for all $1\leq i \leq k$, there is a real
  root $\vartheta$ of $v$ such that when replacing $t$ by $\vartheta$ in
  $\rho_i$, the resulting polynomial has all its roots real. To do this we
  proceed by performing the following steps for $1\leq i\leq k$:
  \begin{enumerate}
  \item first we compute the Sturm-Habicht sequence associated to $\left(
      \rho_i, \frac{\partial \rho_i}{\partial u} \right)$ in $\K[t]$ (the
    Sturm-Habicht sequence is a signed subresultant sequence, see \cite[Chap.
    9, Algo. 8.21]{BPR06});
  \item next, we compute Thom-encodings of the real roots of $v$, which is a way
    to uniquely determine the roots of a univariate polynomial with real
    coefficients by means of the signs of its derivatives at the considered real
    root (see e.g. \cite[Chap. 10, Algo. 10.14]{BPR06});
  \item finally, for each real root $\vartheta$ of $v$, evaluate the signed
    subresultant sequence at $\vartheta$ \cite[Chap. 10, Algo. 10.15]{BPR06} and
    compute the associated Cauchy index to deduce the number of real roots of
    $\rho_i$ (see \cite[Cor. 9.5]{BPR06}). 
  \end{enumerate}
\item[(iv)]  For a given real root $\vartheta$ of $v$, it holds that, for all
  $1\leq i\leq k$, the number of real roots of $\rho_i$ equals its degree, if
  and  only if $V_\lambda$ is non-empty.
\end{itemize}

The above steps describe our {\sf Decide}, which returns {\sf false}
if $V_\lambda$ contains real points, else {\sf true}.

\subsection{The main algorithm}

Our main algorithm ${\sf Real\_emptiness}$ takes symmetric polynomials
$\bm{f} = (f_1, \dots, f_s)$ in $\K[x_1, \dots, x_n]$, with $s<n$, which
satisfy condition $({\sf A})$, and decides whether $V_\R(f)$ is empty.

For a partition $\lambda$, we first find the polynomials $\bm{f}_\lambda
:= \T_\lambda(\bm{f})$, which are $S_\lambda$-invariant in $\K[\bm{z}_1,
  \dots, \bm{z}_k]$, where $\T_\lambda$ is defined as in
\eqref{eq:operator_T}. By Corollary~\ref{lemma:extend_slice},
$\bm{f}_\lambda$ satisfies condition {\sf (A)}, so we can apply the
results of Section~\ref{sec:locus}.

Let $\calN_\fraka$ be the map defined in \eqref{eq:calNN} and
$\calA_\lambda \subset \Kbar{}^{\macrodelta_1} \times \cdots \times
\Kbar{}^{\macrodelta_1}$ be the non-zero Zariski open set defined in
Proposition \ref{prop:feasible}. Assume $\a$ is chosen in
$\calA_\lambda$ (this is one of the probabilistic aspects of our
algorithm) at step \ref{step:n}. By
Corollary~\ref{lemma:extend_slice}, $\bm{f}_\lambda$ satisfies condition
$({\sf A})$. Then, the critical locus of the restriction of $\calN_{\a}$
to $V(\bm{f}_{\lambda})$ is of dimension at most zero (by Proposition
\ref{prop:feasible}). In addition, the map $\calN_{\a}$ is invariant under
the action of the group $S_{\lambda}$.

Let $ \zeta_{\calN_{\a}} $ 
and $\zeta_{\bm{f}_{\lambda}}$ in $\K[\bm{e}_1,
\dots, \bm{e}_k]$ such that  
\[
 \calN_{\a} = \zeta_{\calN_{\a}}(\bm{E}_1, \dots, \bm{E}_k) \quad
 {\rm and} \quad  \bm{f}_{\lambda}=  \zeta_{\bm{f}_{\lambda}}(\bm{E}_1, \dots,
 \bm{E}_k). \]   
Here $\bm{E}_i = (E_{1, i}, \dots, E_{t_i, i})$ denotes the
vector of elementary symmetric polynomials in variables $\bm{z}_i$. In the
next step, we compute a zero-dimensional parametrization
$\scrR_{\lambda}$  of the critical set $W_{\lambda}:= W(\zeta_{\calN_{\a}},
\zeta_{\bm{f}_{\lambda}})$ of $\zeta_{\calN_{\a}}$ restricted to
$V(\zeta_{\bm{f}_{\lambda}})$ by using the  ${\sf 
  Critical\_points}$ algorithm from Lemma~\ref{lemma:subroutine}.  The
parametrization $\scrR_{\lambda}$ is
given by a sequence of polynomials $(v, v_{1,1}, \dots,
v_{t_1, 1}, \dots, v_{1, k}, \dots, v_{t_k, k})$ in $\K[t]$  and a
linear form $\mu$.   

At the final step, we run the ${\sf Decide}(\scrR_{\lambda})$ in order
to determine whether the preimage of $W_{\lambda}$ by the map
$F_{\lambda}$ contains real points. 

\begin{algorithm}[!h]
\caption{${\sf Real\_emptiness}(\bm{f})$}
\begin{itemize}
\item[] \hspace{-1cm} {\bf Input}: symmetric polynomials $\bm{f} = (f_1,
\dots, f_s)$ in $\K[x_1, \dots, x_n]$ with $s < n$ such that $\bm{f}$ satisfies
$({\sf A})$ 
\item[] \hspace{-1cm}  {\bf Output}: {\sf false} if $V(\bm{f}) \cap \R^n$ is non-empty; {\sf
  true} otherwise 
\end{itemize}
\noindent\rule{8.5cm}{0.5pt}
\begin{enumerate}
\item for all partitions $\lambda = (n_1^{t_1}\, \dots \, n_k^{t_k})$
  of $n$ of length at least $s$, do
\begin{enumerate}
\item compute $\bm{f}_\lambda = \T_\lambda(\bm{f})$, where $\T_\lambda$ is defined in
  \eqref{eq:operator_T}\label{step:1}
\item using a chosen $\a \in \calA$, where $\calA$ is defined as in Prop \ref{prop:feasible} \label{step:n},  we construct $\calN_\fraka$ as in \eqref{eq:calNN} and then compute $\calN_{\a}$ 
\item compute $\scrR_\lambda = {\sf
    Critical\_points}(\calN_{\a},  \bm{f}_\lambda)$  \label{step:real_part} 
\item run ${\sf Decide}(\scrR_\lambda)$ \label{step:decide}
\item if  ${\sf Decide}(\scrR_\lambda)$ is {\sf false} return {\sf
    false} \label{step:retfalse}
\end{enumerate}
\item return {\sf true}.
\end{enumerate}
\label{main_algo}
\end{algorithm}

\begin{proposition}
  Assume that, on input symmetric $\bm{f}$ as above,
  and satisfying condition
  {\sf (A)}, for all partitions $\lambda$ of length at least $s$, $\a$ is chosen
  in $\calA_\lambda$ and that all calls to the randomized algorithm {\sf
    Critical\_points} return the correct result.
  Then Algorithm ${\sf Real\_emptiness}$ returns {\sf true} if $V(\bm{f})\cap\R^n$ is
  empty and otherwise it returns {\sf false}. 
\end{proposition}

\begin{proof}
  Since $\bm{f}$ satisfies condition $({\sf A})$, Lemma \ref{lemma:extend_slice}
  implies that $\bm{f}_\lambda$ also satisfies this condition. Then, by the Jacobian
  criterion \cite[Thm 16.19]{eisenbud2013commutative}, $V(\bm{f}_\lambda)$ is smooth
  and equidimensional of dimension $(\ell-s)$, where $\ell$ is the length of
  $\lambda$. Therefore, if $\ell < s$, then the algebraic set $V(\bm{f}_\lambda)$ is
  empty. Thus, the union of $V(\bm{f}_\lambda)\cap \calU_\lambda$ where
  $\calU_\lambda$ is the open set defined in
  Subsection~\ref{ssec:finitestatement} and $\lambda$ runs over the partitions
  of $n$ of length at least $s$, forms a partition of $V(\bm{f})$. Hence,
  $V(\bm{f})\cap\R^n$ is non-empty if and only if there exists at least one such
  partition for which $V(\bm{f}_\lambda)\cap \calU_\lambda\cap \R^n$ is non-empty.

  We already observed that for all $\lambda$, $\bm{f}_\lambda$ does
  satisfy condition $({\sf A})$. 
  Since we have assumed that each time Step
  \ref{step:n} is performed, $\a$ is chosen in $\calA_\lambda$, we apply
  Proposition~\ref{prop:newpols} to deduce that the conditions of
  Lemma~\ref{lemma:subroutine} are satisfied. Hence, all calls to 
  {\sf Critical\_points} are valid.

  Note that since we assume that all these calls return the correct result, we
  deduce that their output encodes points which all lie in $V(\bm{f})$. Hence, if
  $V(\bm{f})\cap\R^n$ is empty, applying the routine {\sf Decide} on these outputs
  will always return {\sf true} and, all in all, our algorithm returns {\sf
    true} when $V(\bm{f})\cap\R^n$ is empty.

  It remains to prove that it returns {\sf false} when $V(\bm{f})\cap\R^n$ is
  non-empty. Note that there is a partition $\lambda$ such that
  $V(\bm{f}_\lambda)\cap\R^n$ is nonempty and has an empty intersection with the
  complement of $\calU_\lambda$. That is, all connected components of
  $V(\bm{f}_\lambda)\cap\R^n$ are in $\calU_\lambda$. 
  
  Let $C$ be such a connected component.  By
  Lemma~\ref{lemma:proper_power}, the map $\calN_{\a}$ is proper, and
  non-negative. Hence, its restriction to $V(\bm{f}_\lambda)\cap \R^n$
  reaches its extremum at all connected components of
  $V(\bm{f}_\lambda)\cap \R^n$.  This implies that the restriction of
  $\calN_{\a}$ to $V(\bm{f}_\lambda)$ has real critical points which are
  contained in $C$ (and by Proposition~\ref{prop:feasible} there are
  finitely many). Those critical points are then encoded by the output
  of the call to {\sf Critical\_points} (Step~\ref{step:real_part})
  and {\sf false} is returned.
\end{proof}

\subsection{Complexity analysis}

 Let $d = \max(\deg(\bm{f}))$.  First for a partition $\lambda$, applying
 $\T_\lambda$ to $\bm{f}$ takes linear time in $O(n {{n+d}\choose d})$,
 the number of monomials of $\bm{f}$ and the cost of Step \ref{step:n} is
 nothing. At the core of the algorithm, computing $\scrR_\lambda$ at
 Step~\ref{step:real_part} requires $\softO\left(
 \delta^2c_\lambda(e_\lambda + c_\lambda^5)n^4\Gamma\right)$
 operations in $\K$ by Lemma~\ref{lemma:subroutine}, where $\delta =
 \max(d, \deg(\phi_{\a}))$. Also, the degree of $\scrR_\lambda$ is at
 most $c_\lambda$.

 In order to determine the cost of the {\sf Decide} process at Step
 \ref{step:decide}, let $a$ be the degree of $v$ and $b$ be the maximum of the
 partial degrees of $\rho_i$'s w.r.t. $u$. By the complexity analysis of
 \cite[Algo. 8.21 ; Sec. 8.3.6]{BPR06}, Step {(1)} above is performed within
 $O\left( b^{4}a \right)$ arithmetic operations in $\K[t]$ using a classical
 evaluation interpolation scheme (there are $b$ polynomials to interpolate, all
 of them being of degree $\leq 2ab$). Step {(2)} above requires $O\left(
   a^{4}\log(a) \right)$ arithmetic operations in $\K$ (see the complexity
 analysis of \cite[Algo 10.14; Sec. 10.4]{BPR06}). Finally, in Step (3), we
 evaluate the signs of $b$ polynomials of degree $\leq 2ab$ at the real roots of
 $v$ (of degree $a$) whose Thom encodings were just computed. This is performed
 using $O\left( b a^{3}\left( (\log(a) + b) \right) \right)$ arithmetic
 operations in $\K$ following the complexity analysis of \cite[Algo 10.15; Sec.
 10.4]{BPR06}. The sum of these estimates lies in $O\left(b^{4}a + b
   a^{4}\left( (\log(a) + b) \right) \right)$.

 Now, recall that the degree of $v$ is the degree of $\scrR_\lambda$, so $a \le
 c_\lambda$. The degree of $\rho_i $ w.r.t. $u$ equals $t_i$ and $t_i\leq n$.
 This means $b \le n$. All in all, we deduce that the total cost of this final
 step lies in $O\left( n^{4}c_\lambda + n^{2}c_\lambda\right)$, which
 is negligible compared  to the previous costs.

In the worst case, one need to consider all the partitions of $n$ of
length at least $s$. Thus the total complexity of ${\sf
Real\_emptiness}$ is 
\[
\sum_{\lambda, \ell \ge s} \softO\left(
  \delta^2c_\lambda(e_\lambda + c_\lambda^5)n^4\Gamma\right) 
\]operations in $\K$.  In addition, Lemma 34 in
\cite{faugere2020computing} implies that $$\sum_{\lambda, \ell \ge s}
c_\lambda \le c \ {\rm and} \ \sum_{\lambda, \ell   \ge s}  e_\lambda
\le e,$$ where  $c = \deg(\zeta_{\bm{f}_\lambda})^s{{n+\delta-1}\choose n}
\ {\rm and}  \ e = n(\deg(\zeta_{\bm{f}_\lambda})+1)^s{{n+\delta}\choose
  n}$. Notice further that  ${{n+\delta}\choose
  \delta} \le (n+1){{n+\delta-1}\choose 
  d}$ and $e =n(d+1)^s{{n+\delta}\choose n} \le
 n(n+1)c^5$ for $\delta\ge
2$. In addition, since $\deg(\phi_{\a}) \le \max(t_i)+1 \le n$,
the total cost of our
algorithm is 
\[
\softO\left(d^2n^6c^6\Gamma\right)
 = \softO\left(d^{6s+2}n^{11}{{n+d}\choose n}^6\left({{n+d}\choose n}
     +  {{n}\choose {s+1}} \right) \right)  
\]
 operations in $\K$.

\subsection{An  example }
Let $n=4$ and $s=1$ with  $\bm{f} = (f)$ where $$f = x_1^2 + x_2^2 +
x_3^2 + x_4^2 -6x_1x_2x_3x_4  -1.$$  Consider first the
partition $\lambda= (4^1)$. Then $f_\lambda := \T_\lambda(f)=
-6z_{1,1}^4+4z_{1,1}^2-1$ which has no real solution as
$f_\lambda = -2z_{1,1}^4-(2z_{1,1}^2-1)^2 < 0$ for all  $z_{1,1} \in
\R$.

Next we consider $\lambda=(2^2)$. Then $$f_{(2^2)} = 2z_{1,1}^2+2z_{2,1}^2
-6z_{1,1}^2z_{2,1}^2-1$$ and we take $\phi = 
5(z_{1,1}^2+z_{2,1}^2)-9(z_{1,2}+z_{2,1})-3$.  In this case
$\zeta_{f_{(2^2)}}~=~2e_{1,1}^2-6e_{2,1}^2-4e_{2,1}-1$ and
$\zeta_{\phi} =
5e_{1,1}^2-9e_{1,1}-10e_{2,1}-3$.
The critical points of $\zeta_{\phi}$
restricted to $V(\zeta_{f_{(2^2)}})$ are solutions to
\[
\zeta_{f_{(2^2)}} = \det\big(\jac(\zeta_{f_{(2^2)}},
\zeta_{\phi})\big) = 0,\] that is $2e_{1,1}^2-6e_{2,1}^2-4e_{2,1}-1 =
120e_{1,1}e_{2,1}-108e_{2,1}-36 = 0
$. A zero-dimensional parametrization of these critical points is
given by 
$((v, v_{1,1}, v_{2,1}), \mu)$, where  
\begin{align*}
& v = 200t^4-360t^3+62t^2+60t-27, \\
& v_{1,1}  = t, {~\rm and~}\\ &v_{2,1} =
-\frac{1}{6}t^3+\frac{9}{20}t^2-\frac{31}{600}t-1/20.
\end{align*}
At the final step, we check that the system 
\[
\rho_1 = v = 0, \quad {\rm with} \quad \rho_1 =
v'u^2-v_{1,1}u+v_{2,1} \in \Q[t,u], 
\] has real solutions. This implies that $V_\R(f)$ is non-empty.

The output of our algorithm is consistent with the fact that the point 
$(1,1,1/2,1/2)$ is in $V_\R(f)$.

\section{Topics for Future Research}

Determining topological properties of a real variety $V_\R(\bm{f})$ is an
important algorithmic problem. Here we have presented an efficient
algorithm to determine if $V_\R(\bm{f})$ is empty or not.  More generally,
we expect that the ideas presented here may lead to algorithmic
improvements also in more refined questions, like computing one point
per connected component or the Euler characteristic for a real
symmetric variety. Furthermore, while our complexity gains are
significant for symmetric input we conjecture that we can do better in
certain cases. In particular, when the degree of the polynomials is at
most $n$ then we expect we that a combination with the topological
properties of symmetric semi algebraic sets found in \cite[Prop
  9]{basu2022vandermonde} can reduce the number of orbits considered,
for example, instead of $n^d$ we might only need $n^{d/2}$ for fixed
$d$. Finally, a generalization to general symmetric semi algebraic
sets should be possible.
\newpage

\balance

\bibliographystyle{ACM-Reference-Format}
\bibliography{biblio}

\end{document}